\pdfoutput=1
\newif\ifAnon
\Anontrue
\newif\ifFull
\Fulltrue

\newif\ifDraft
\Draftfalse

\ifFull
\Anonfalse
\documentclass[11pt]{article}
\advance \topmargin by -\headheight
\advance \topmargin by -\headsep
\evensidemargin \oddsidemargin
\else
\documentclass[graphicx]{sig-alternate}
\conferenceinfo{CODASPY'12,} {February 7--9, 2012. San Antonio, TX, USA.}
\CopyrightYear{2012}
\crdata{978-xxx}
\fi

\usepackage{graphicx}
\usepackage{verbatim}
\usepackage{algorithm}
\usepackage{algorithmic}
\usepackage{amsmath}
\usepackage{amssymb}
\usepackage{multirow}
\usepackage{array}
\usepackage{times}
\usepackage{url}

\makeatletter
\def\@begintheorem#1#2{\sl \trivlist \item[\hskip \labelsep{\bf #1\ #2:}]}
\def\@opargbegintheorem#1#2#3{\sl \trivlist
      \item[\hskip \labelsep{\bf #1\ #2\ #3:}]}
\makeatother

\newtheorem{theorem}{Theorem}
\newtheorem{lemma}[theorem]{Lemma}

\ifFull
\newenvironment{proof}{\noindent{\bf Proof:}}{\hspace*{\fill}\rule{6pt}{6pt}\bigskip} 
\fi

\newcommand{\ignore}[1]{}
\newcommand{\msgsize}{M}
\newcommand{\memconst}{b}
\newcommand{\memsize}{\memconst\msgsize}

\ifDraft
\newcommand{\comments}[1]{\textsf{#1}}
\newcommand{\olya}[1]{\textbf{olya: #1}}
\else
\newcommand{\comments}[1]{}
\newcommand{\olya}[1]{}
\fi

\begin{document}

\ifFull\else
\pagestyle{plain}
\def\thepage{\arabic{page}}
\fi

\ifFull
\title{Oblivious Storage with Low I/O Overhead}
\else
\title{Practical Oblivious Storage}
\fi

\ifAnon
\author{Anonymous submission to CODASPY 2012}
\else
\author{
{Michael T. Goodrich} \\
Dept.~of Computer Science \\ 
University of California, Irvine \\
\vspace{0.1in}
goodrich@ics.uci.edu\\
Michael Mitzenmacher \\
Dept.~of Computer Science \\ 
Harvard University \\
michaelm@eecs.harvard.edu\\
\and
Olga Ohrimenko \\
Dept.~of Computer Science \\
Brown University \\
\vspace{0.1in}
olya@cs.brown.edu \\
Roberto Tamassia \\
Dept.~of Computer Science \\
Brown University \\
rt@cs.brown.edu
}
\fi

\ifFull
\date{}
\else
\date{~\vspace*{-1in}}
\fi
\maketitle

\begin{abstract}
  We study \emph{oblivious storage} (OS), a natural way to model
  privacy-preserving data outsourcing where a client, Alice, stores
  sensitive data at an honest-but-curious server, Bob.  We show that
  Alice can hide both the content of her data and the pattern in which
  she accesses her data, with high probability, using a method that
  achieves $O(1)$ amortized rounds of communication between her and
  Bob for each data access.  We assume that Alice and Bob exchange
  small messages, of size $O(N^{1/c})$, for some constant $c\ge2$, in
  a single round, where $N$ is the size of the data set that Alice is
  storing with Bob.  We also assume that Alice has a private memory of
  size $2N^{1/c}$.  These assumptions model real-world cloud storage
  scenarios, where trade-offs occur between latency, bandwidth, and
  the size of the client's private memory.
\end{abstract}

\section{Introduction}
Outsourced data management is a large and growing industry.
For example, as of July 2011, Amazon S3~\cite{amazon-s3} reportedly
stores more than 400 billion objects, which is four times 
its size from the year before, and
the Windows Azure service~\cite{azure}, which was started in late
2008, is now a multi-billion dollar enterprise.

With the growing impact of online cloud storage technologies, there is
a corresponding growing interest in methods for privacy-preserving
access to outsourced data.  Namely, it is anticipated
that many customers of
cloud storage services will desire or require that their data remain private.
A necessary component of private data access, of course, is to
encrypt the objects being stored.  But information can be leaked from
the way that data is accessed, even if it is
encrypted (see, e.g.,~\cite{cwwz-sclwa-10}).  Thus, privacy-preserving data access
must involve both encryption and techniques for obfuscating the
patterns in which users access data.

\paragraph{Oblivious RAM Simulation}
One proposed approach to privacy-preserving data access involves
\emph{oblivious random access machine} (\emph{ORAM}) simulation~\cite{go-spsor-96}.
In this approach, the client, Alice, is modeled as a CPU with 
a limited-size cache that accesses a large
indexed memory managed by the owner of the data service, Bob.
The goal is for Alice to perform an arbitrary RAM computation while
completely obscuring from Bob the data items she accesses and the access pattern.
Unfortunately, although known ORAM simulations~\cite{%
a-orwca-10,%
dmn-psor-10,%
gm-paodor-11,%
go-spsor-96,%
gmot-orsew-11,%
pr-orr-20,%
shornbs-klo-11,%
sss-tpor-11,%
DBLP:conf/ndss/WilliamsS08}
can be adapted to the problem of privacy-preserving access to
outsourced data, they do not naturally match the interfaces provided
by existing cloud storage services, which are not organized according
to the RAM model (e.g., see~\cite{bmp-rosmor-11}).

\paragraph{Oblivious Storage}
A notable exception to this aspect of previous work on ORAM simulation
is a recent paper by  Boneh {\it et al.}~\cite{bmp-rosmor-11}, who
introduce the \emph{oblivious storage} (\emph{OS}) model.
In this model, the storage provided by Bob is viewed more
realistically as a collection of key-value pairs and the query 
and update operations supported by his API are
likewise more accurately viewed in 
terms of operations dealing with key-value pairs, which we 
also call \emph{items}.
An OS solution is \emph{oblivious} in this context if an honest-but-curious 
polynomial-time adversary is unable to distinguish between the (obfuscated)
versions of two possible access sequences of equal 
length and maximum set size,
which are polynomially related,
beyond a negligible probability.
Although the solution to the OS problem given by Boneh {\it et al.}~is 
somewhat complicated, it is
nevertheless considerably simpler than most 
of the existing ORAM solution techniques.
In particular, it avoids additional details required of ORAM simulations that 
must deal with the obfuscation of an arbitrary RAM algorithm.
Thus, an argument can be made that 
the OS approach is both more realistic and supports simpler oblivious
simulations.
The goal of this paper, then, is to explore further simplifications and
improvements 
to achieve practical solutions to the oblivious storage problem.

\subsection{Related Previous Work}
\begin{table*}
\ifFull
\small
\fi
\begin{center}
{
\renewcommand{\arraystretch}{1.15}
\ifFull
\begin{tabular}{>{\centering}m{2.77cm}|>{\centering}m{2.6cm}|>{\centering}m{2.6cm}|>{\centering}m{1.7cm}|c|c}
\multirow{2}{*}{Method} & \multicolumn{2}{c|}{Access Overhead} &  \multicolumn{1}{c|}{Message Size} &  \multicolumn{1}{c|}{Client } &   \multirow{2}{*}{Server Storage} \\
\cline{2-3}
 & Online & Amortized & ($M$)& Memory & \\
 \else
 \begin{tabular}{c|c|c|c|c|c}
 \multirow{2}{*}{Method} & \multicolumn{2}{c|}{Access Overhead} &  \multirow{2}{*}{Message Size ($M$)} &  \multirow{2}{*}{Client Memory} &   \multirow{2}{*}{Server Storage} \\
\cline{2-3}
 & Online & Amortized & &  & \\
 \fi
\hline
Shi {\it et al.}~\cite{scsl-orwcc-11} & $O(\log^3 N)$ & $O(\log^3 N)$ & $O(1)$ & $O(1)$ & $O(N\log N)$\\
Williams {\it et al.}~\cite{wsc-bcomp-08} & $O(\log N \log \log N)$ & $O(\log N \log \log N)$ & $O(N^{1/2})$ & $O(N^{1/2})$ & $O(N)$\\
Goodrich {\it et al.}~\cite{gmot-orsew-11} & $O(\log N)$ & $O(\log N)$ & $O(1)$ & $N^\nu$ & $O(N)$\\
Boneh {\it et al.}~\cite{bmp-rosmor-11} & $O(1)$ &$O(\log N)$ & $N^{1/2}$ & $O(N^{1/2})$& $N + 2N^{1/2}$
\tabularnewline
\hline 
Our Method $N^{1/2}$ & $O(1)$ & $O(1)$&  $N^{1/2}$ &  $2N^{1/2}$ & $N + N^{1/2}$\\
Our Method $N^{1/c}$ & $O(1)$ &$O(1)$ &  $N^{1/c}$ &  $cN^{1/c}$ & $N+2\sum_{i=1}^{c-2}N^{(c-i)/c}$\\
\end{tabular}
}
\end{center}
\caption{\label{tbl:results}%
  Comparison between selected oblivious storage approaches where
  online access overhead is the number of accesses required to retrieve the requested
  item.
 Here, $N$ denotes the number of items, and $0<\nu<1$ 
and $c\ge 2$ are arbitrary positive constants.
  The message size, client memory, and server 
storage are measured in terms of the number of items.
Also, we note that the constant factor in the $O(1)$ access
overhead for our $N^{1/c}$ inductive method depends on 
the constant $c$.
 }
\end{table*}
Research on oblivious simulation of one computational model by another began 
with Pippenger and Fischer~\cite{pf-racm-79}, who show that one can
obliviously simulate a computation of a 
one-tape Turing machine computation of length $N$ with
an two-tape Turing machine computation of length $O(N\log N)$. 
That is, they show how to perform such an oblivious simulation 
with a computational overhead that is $O(\log N)$.

Goldreich and Ostrovsky~\cite{go-spsor-96} show that one can perform an 
oblivious RAM (ORAM) simulation using an outsourced data
server and they prove a lower bound implying
that such simulations require an overhead 
of at least $\Omega(\log N)$, for a RAM memory of size $N$,
under some reasonable assumptions about the nature of such simulations.
For the case where Alice has only a constant-size private memory,
they show how Alice can easily achieve an overhead of $O(N^{1/2}\log N)$,
using a scheme called the ``square-root solution,''
with $O(N)$ storage at Bob's server. 
With a more complicated scheme, 
they also show how Alice can achieve an overhead of
$O(\log^3 N)$ with $O(N\log N)$ storage at Bob's server, using a scheme
called the ``hierarchical solution.''

Williams and Sion~\cite{DBLP:conf/ndss/WilliamsS08}
provide an ORAM simulation for the case when the data
owner, Alice, has a private memory of size $O({N}^{1/2})$.
They achieve an expected amortized time overhead of $O(\log^2 N)$ using
$O(N\log N)$ memory at the external data provider, Bob.
Additionally, Williams {\it et al.}~\cite{wsc-bcomp-08} claim a
result that uses an $O({N}^{1/2})$-sized private memory
and achieves $O(\log N\log\log N)$ amortized time overhead with
a linear-sized outsourced storage, but some
researchers (e.g., see~\cite{pr-orr-20}) 
have raised concerns with the assumptions and analysis of
this result.
Likewise,
Pinkas and Reinman~\cite{pr-orr-20} published an ORAM simulation
result for the case where Alice maintains a 
constant-size private memory, claiming that Alice can achieve
an expected amortized
overhead of $O(\log^2 N)$ while using $O(N)$ storage space,
but Kushilevitz {\it et al.}~\cite{shornbs-klo-11} 
have raised correctness issues with this result as
well~\cite{shornbs-klo-11}.
Goodrich and Mitzenmacher~\cite{gm-paodor-11} 
show that one can achieve an overhead of
$O(\log^2 N)$ in an ORAM simulation, with high probability,
for a client with constant-sized
local memory, and $O(\log N)$, for a client with $O(N^\epsilon)$
memory, for a constant $\epsilon>0$.
Kushilevitz {\it et al.}~\cite{shornbs-klo-11} 
also show that one can achieve
an overhead of 
$O(\log^2 N/\log\log N)$ in an ORAM simulation, with high probability,
for a client with constant-sized local memory.
Ajtai~\cite{a-orwca-10} proves that ORAM
simulation can be done with polylogarithmic overhead 
without cryptographic assumptions about the existence of
random hash functions, as is done in the papers mentioned above
(and this paper), and
a similar result is given by Damg\aa{}rd {\it et al.}~\cite{dmn-psor-10}.

The importance of privacy protection in outsourced data management 
naturally raises
the question of the practicality of the previous ORAM solutions.
Unfortunately, the above-mentioned theoretical results contain several
complications and hidden constant factors that make these solutions less than
ideal for real-world use.
Stefanov {\it et al.}~\cite{sss-tpor-11}
study the ORAM simulation problem from a practical point of view,
with the goal of reducing the worst-case bounds for data accesses.
They show that one can achieve an amortized overhead of $O(\log N)$ 
and worst-case performance $O(N^{1/2})$,
with $O(\epsilon N)$ storage on the client, for a constant $0<\epsilon<1$,
and an amortized overhead of $O(\log^2 N)$ and similar worst-case
performance,
with a client-side storage of
$O(N^{1/2})$, both of which have been hidden constant factors than previous
ORAM solutions.
Goodrich {\it et al.}~\cite{gmot-orsew-11}
similarly study methods for improving the worst-case performance of ORAM
simulation, showing that one can achieve a worst-case overhead of $O(\log N)$
with a client-side memory of size $O(N^\epsilon)$, for any constant
$\epsilon>0$.

As mentioned above,
Boneh {\it et al.}~\cite{bmp-rosmor-11}
introduce the \emph{oblivious storage} (OS) problem and argue how it is 
more realistic and natural than the ORAM simulation problem.
They study methods that separate access overheads and the overheads needed for 
rebuilding the data structures on the server, providing, for example, $O(1)$
amortized overhead for accesses with $O((N\log N)^{1/2})$ overhead for
rebuilding operations, assuming a similar bound for the 
size of the private memory on the client.

\subsection{Our Results}
In this paper, we study the oblivious storage (OS) problem, providing 
solutions that are parameterized by the two critical components
of an outsourced storage system:
\begin{itemize}
\item
$N$: the number of items that are stored at the server
\item
$M$: the maximum 
number of items that can be sent or received in a single message,
which we refer to as the \emph{message size}.
\end{itemize}

We assume that the objects being outsourced to Bob's cloud storage are
all of the same size, since this is a requirement to achieve oblivious
access. Thus, we can simply refer to the memory and message sizes in
terms of the number of items that are stored.
This notation is borrowed from the literature on external-memory algorithms
(e.g., see~\cite{DBLP:reference/algo/Vitter08}), 
since it closely models the
scenario where the memory needed by a computation exceeds its local capacity
so that external storage is needed.
In keeping with this analogy to external-memory algorithms, we refer to each
message that is exchanged between Alice and Bob as an \emph{I/O}, 
each of which, as noted above, is of size at most~$M$.
We additionally assume that Alice's memory is of size at least~$2M$, 
so that she can hold two messages in her local memory.
In our case, however, we additionally assume that $M \ge N^{1/c}$,
for some constant~$c\ge 2$.
This assumption is made for the sake of realism, since even with 
$c=3$, we can model Bob 
storing exabytes for Alice, while she and he exchange 
individual messages measured in megabytes.
Thus, we analyze our solutions in terms of the constant
\[
c=\log_M N .
\]
We give practical solutions to the oblivious storage problem that
achieve an efficient amortized number of I/Os exchanged between Alice
and Bob in order to perform \textsf{put} and \textsf{get} operations.

We first present a simple ``square-root'' solution, which assumes that $M$ is
$N^{1/2}$, so~$c=2$. 
This solution is not oblivious, however, if the client
requests items that are not in the set. 
So we show how to convert any 
oblivious storage solution that cannot tolerate requests for missing items to a solution that can support obliviously also such requests.
With these tools in hand, we then show how to define an inductive solution to
the oblivious storage problem that achieves a constant amortized number
of I/Os for each access, assuming $M = N^{1/c}$.
We believe that $c=2$, $c=3$, and $c=4$
are reasonable choices in practice, depending on the relative sizes of $M$ and $N$.

The operations in these solutions are factored into \emph{access
  operations} and \emph{rebuild operations}, as in the approach
advocated by Boneh {\it et al.}~\cite{bmp-rosmor-11}.  Access
operations simply read or write individual items to/from Bob's storage
and are needed to retrieve the requested item, whereas rebuild
operations may additionally restructure the contents of Bob's storage
so as to mask Alice's access patterns.  In our solutions, access
operations use messages of size $O(1)$ while messages of size $M$
are used only for rebuild operations.


An important ingredient in all oblivious storage and oblivious RAM
solutions is a method to obliviously ``shuffle'' a set of elements so
that Bob cannot correlate the location of an element before the
shuffle with that after the shuffle.  This is usually done by using an
oblivious sorting algorithm, and our methods can utilize such an
approach, such as the external-memory oblivious sorting algorithm of
Goodrich and Mitzenmacher~\cite{gm-paodor-11}.

In this paper, we also introduce a new simple shuffling method,
which we call the \emph{buffer shuffle}.
We show that this method can shuffle with high probability with very little
information leakage, which is likely to be sufficient in practice in most
real-world oblivious storage scenarios.
Of course, if perfectly oblivious shuffling is desired, then this shuffle
method can be replaced by external-memory sorting, which increases the I/O
complexity of our results by at most a constant factor (which depends
on~$c$).

In Table~\ref{tbl:results}, we summarize our results and compare the
main performance measures of our solutions with those of selected
previous methods that claim to be practical. 

\subsection{Organization of the Paper}

The rest of this paper is organized as follows.  In
Section~\ref{sec:oblivious}, we overview the oblivious storage model
and its security properties and describe some basic techniques used in
previous work. Our buffer shuffle method is presented and analyzed in
Section~\ref{sec:buffer-shuffle}. We give a preliminary
miss-intolerant square-root solution in
Section~\ref{sec:square}. Section~\ref{sec:miss-tolerance} derives a
miss-tolerant solution from our square-root solution using a cuckoo
hashing scheme. In Section~\ref{sec:induction}, we show how to reduce
the storage requirement at the client.  Finally, in
Section~\ref{sec:performance}, we describe our experimental results
and provide estimates of the actual time overhead and monetary cost of
our method, obtained by a prototype implementation and simulation of
the use of our solution on the Amazon S3 storage service.


\section{The Oblivious Storage Model}
\label{sec:oblivious}
In this section, we discuss the OS model using the formalism of Boneh
{\it et al.}~\cite{bmp-rosmor-11}, albeit with some minor
modifications.  As mentioned above, one of the main differences
between the OS model and the classic ORAM model is that the storage
unit in the OS model is an \emph{item} consisting of a key-value
pair. Thus, we measure the size of messages and of the storage space
at the client and server in terms of the number of items

\subsection{Operations and Messages}
Let $S$ be the set of data items.  The server supports the following
operations on~$S$.
\begin{itemize}
\item \textsf{get}$(k)$: if $S$ contains an item, $(k,v)$, with key
  $k$, then return the value, $v$, of this item, else
  return~\textsf{null}.
\item \textsf{put}$(k,v)$: if $S$ contains an item, $(k,w)$, with key
  $k$, then replace the value of this item with $v$, else add to $S$ a
  new item~$(k,v)$.
\item \textsf{remove}$(k)$: if $S$ contains an item, $(k,v)$, with key
  $k$, then delete from $S$ this item and return its value $v$, else
  return~\textsf{null}.
\item \textsf{getRange}$(k_1,k_2,m)$: return the first $m$ items (by
  key order) in $S$ with keys in the range $[k_1, k_2]$.  Parameter $m$
  is a cut-off to avoid data overload at the client because of an
  error.  If there are fewer than $m$ such items, then all the items with
  keys in the range $[k_1,k_2]$ are returned.
\item \textsf{removeRange}$(k_1,k_2)$: remove from $S$ all items with
  keys in the range $[k_1,k_2]$. 
\end{itemize}

The interactions between the client, Alice, and the server, Bob, are
implemented with messages, each of which is of size at most
$\msgsize$, i.e., it contains at most $\msgsize$ items.
Thus, Alice can send Bob a single message consisting
of $\msgsize$ \textsf{put} operations, each of which adds a single item.
Such a message would count as a single I/O.
Likewise, the response to a \textsf{getRange}$(k_1,k_2,m)$ operation
requires $O(\lceil m/\msgsize\rceil)$ I/Os; 
hence, Alice may wish to limit $m$ to be
$O(\msgsize)$. 
Certainly, Alice would want to limit $m$ to be $O(M)$ in most cases,
since she would otherwise
be unable to locally store the entire result of such a query if it
reaches its cut-off size.

As mentioned above, our use of parameter $\msgsize$ is done for the sake of
practicality, since it is unreasonable to assume that Alice and Bob can only
communicate via constant-sized messages.
Indeed, with network connections measured in gigabits per second but with
latencies measured in milliseconds, the number of rounds of communication is
likely to be the bottleneck, not bandwidth.
Thus, because of this orders-of-magnitude difference
between bandwidth and latency, we assume
\[
\msgsize \ge N^{1/c},
\]
for some fixed constant $c\ge 2$, but that Alice's memory is smaller than~$N$.
Equivalently, we assume that $c=\log_M N$ is a constant.
For instance, as highlighted above, if Bob's memory is measured 
in exabytes and we take $c=3$, then we are
reasonably assuming that Alice and Bob can exchange messages whose sizes are
measured in megabytes.
To assume otherwise would be akin to trying to manage a 
large reservoir with a pipe the size of a drinking straw.

We additionally
assume that Alice has a private memory of size $\memsize$, in which she can 
perform computations that are hidden from the server, Bob.
To motivate the need for Alice outsourcing her data, while 
also allowing her to
communicate effectively with Bob, we assume that $b\ge 2$
and $2\msgsize < N$.

\subsection{Basic Techniques}
Our solution employs several standard techniques previously introduced
in the oblivious RAM and oblivious storage literature.  To prevent Bob
from learning the original keys and values and to make it hard for Bob
to associate subsequent access to the same item, Alice replaces the
original key, $k$, of an item with a new key $k'=h(r || k)$, where $h$
is a cryptographic hash function (i.e., one-way and
collision-resistant) and $r$ is a secret randomly-generated nonce that
is periodically changed by Alice so that a subsequent access to the
same item uses a different key. Note that Bob learns the modified keys
of the items. However, he cannot derive from them the original keys
due to the one-way property of the cryptographic hash function
used. Also, the uniqueness of the new keys occurs with overwhelming
probability due to collision resistance.

Likewise, before storing an item's value, $v$, with Bob, Alice
encrypts $v$ using a probabilistic encryption scheme. E.g., the
ciphertext is computed as $E(r || v)$, where $E$ is a deterministic
encryption algorithm and $r$ is a random nonce that gets discarded
after decryption. Thus, a different ciphertext for $v$ is generated
each time the item is stored with Bob. As a consequence, Bob cannot
determine whether $v$ was modified and cannot track an item by its
value.  The above obfuscation capabilities are intended to make it
difficult for Bob to correlate the items stored in his memory at
different times and locations, as well as make it difficult for Bob to
determine the contents of any value.

We distinguish two types of OS solutions.  We say that an oblivious
storage solution is \emph{miss-intolerant} if it does not allow for
\textsf{get} requests that return \texttt{null}. Thus, Alice must know
in advance that Bob holds an item with the given key.  In applications
that by design avoid requests for missing items, this restriction
allows us to design an efficient oblivious-storage solution, since we
don't have to worry about any information leakage that comes from
queries for missing keys.  Alternatively, if an oblivious storage
solution is oblivious even when accesses can be made to keys that are
not in the set, then we say that the solution is
\emph{miss-tolerant}.

\subsection{Security Properties}


Our OS solution is designed to satisfy the following security
properties, where the adversary refers to Bob (the server) or a third
party that eavesdrops the communication between Alice (the client) and
Bob. The adversary is assumed to have polynomially bounded
computational power.

\begin{description}
\item[Confidentiality.] Except with negligible probability, the
  adversary should be unable to determine the contents (key or value)
  of any item stored at the server. This property is assured by the
  techniques described in the previous subsection.

\item[Hardness of Correlation.]  
  Except with negligible or very low probability beyond $1/2$,
  the adversary should be unable to
  distinguish between any two possible access
  sequences of equal length and maximum set size.
  That is, consider two possible access sequences, $\sigma_1$
  and $\sigma_2$, that consist of $L$ operations, \textsf{get}, 
  \textsf{put}, and \textsf{remove},
  that could be made by Alice, on a set of size up to $N$, where
  $L$ is polynomial in $N$.
  Then an oblivious storage (OS) solution has \emph{correlation hardness}
  if it applies an obfuscating 
  transformation so that, after seeing the sequence of I/Os 
  performed by such a transformation, the probability that
  Bob can correctly guess whether Alice has performed
  (the transformed version of)
  $\sigma_1$ or $\sigma_2$ is
  more than $1/2$ by at most a $1/N^\alpha$
  or a negligible amount, 
  depending on the degree of obfuscation desired, 
  where $\alpha>1$ is a constant.\footnote{We assume $L << N^\alpha$
  in this case.}
\end{description}

Note that $N$ is used in the definition of ``correlation hardness''
in both the upper bound 
on the size of Alice's set and also in the probability of 
Bob correctly
distinguishing between two of her possible access sequences.
Of course,
the efficiency of an OS solution should also to be measured in terms of $N$.

\section{The Buffer Shuffle Method}
\label{sec:buffer-shuffle}

One of the key techniques in our solutions is the use of 
oblivious shuffling.
The input to any \emph{shuffle}
operation is a set, $A$, of $N$ items.  
Because of the inclusion of the \textsf{getRange} operation in the server's
API, we can view the items in $A$ as being ordered by their keys.
Moreover, this functionality also allows us to
access a contiguous run of $M$ such items, starting from a given key.
The output of a shuffle
is a reordering of the items in $A$ with replacement keys,
so that all permutations are equally likely.
During a shuffle, the server, Bob, can observe Alice read (and remove) $M$
of the items he is storing for her, and then write back $M$ more items, which
provides some degree of obfuscation of how the items in these
read and write groups are correlated.  An additional desire for the output of a shuffle is that,
for any item $x$ in the input, the adversary should be able to correlate $x$ with
any item $y$ in the output only with probability that is very close to $1/N$
(which is what he would get from a random guess).


During such a shuffle, we assume that Alice is wrapping each 
of her key-value pairs, $(k,v)$, as $(k',(k,v))$, where $k'$ 
is the new key that is chosen to obfuscate $k$.
Indeed, it is likely that in each round of communication that Alice makes
she will take a wrapped (input) pair, $(k', X)$,
and map it to a new (output) pair, $(k'',X')$, where the $X'$ is assumed to be 
a re-encryption of $X$.
The challenge is to define an encoding strategy for the $k'$ and $k''$
wrapper keys so that it is difficult for 
the adversary to correlate inputs and outputs.

\subsection{Theoretical Choice: Oblivious Sorting}
One way to do this is to assign each item a random key 
from a very large universe, which is separate and
distinct from the key that is a part of this key-value pair, and
obliviously sort~\cite{gm-paodor-11} the items by these keys. 
That is, we can wrap each key-value pair, $(k,v)$, as $(k',(k,v))$, where $k'$ 
is the new random key, and then wrap these wrapped pairs in a way that allows
us to implement an oblivious sorting algorithm in the OS model based on
comparisons involving the $k'$ keys.
Specifically,
during this sorting process, we would further wrap each 
wrapped item, $(k',(k,v))$, as $(\alpha,(k',(k,v)))$, 
where $\alpha$ is an address or
index used in the oblivious sorting algorithm.
So as to distinguish such keys even further, Alice can also add a prefix to each
such $\alpha$, such as ``\texttt{Addr:}'' or ``\texttt{Addr}$i$:'', where $i$
is a counter (which could, for instance, be counting the steps in Alice's
sorting algorithm).
Using such addresses as ``keys'' allows Alice to consider Bob's storage as if
it were an array or the memory of a RAM.
She can then use this scheme to simulate an oblivious sorting algorithm.

If the 
randomly assigned keys are distinct, which will occur with 
very high probability, then
this achieves the desired goal. And 
even if the new keys are not distinct, we can
repeat this operation until we get a set of distinct new keys without 
revealing any data-dependent information to the server.

From a theoretical
perspective, it is hard to beat this solution. It is well-known, 
for instance, that shuffling by sorting items via randomly-assigned keys
generates a random permutation such that all permutations are
equally likely (e.g., see~\cite{k-sa-98}).
In addition, since the means to go from the input to the output is
data-oblivious with respect to the I/Os (simulated using the address keys), 
the server who is watching the inputs and
outputs cannot correlate any set of values. 
That is, independent of the set of I/Os,
any input item, $x$, at the beginning of the sort can be 
mapped to any output item, $y$,
at the end.
Thus,
for any item $x$ in the input, the adversary can correlate $x$ with
any item $y$ in the output with probability exactly $1/N$.
Finally, we can use 
the external-memory deterministic oblivious-sorting algorithm
of Goodrich and Mitzenmacher~\cite{gm-paodor-11}, for instance,
so as to use messages of size $O(M^{1/2})$, which will result in an algorithm
that sorts in $O((N/M)\log_{\sqrt{M}} (N/M))=O((N/M)c^2)$ I/Os.
That is, such a sorting algorithm uses a constant amortized number
of I/Os per item.


But using an oblivious sorting algorithm
requires a fairly costly overhead, as the constant
factors and details of this algorithm are somewhat nontrivial.
Thus,
it would be nice in applications that don't necessarily require perfectly
oblivious shuffling to have a simple substitute that could be fast and
effective in practice.

\subsection{The Buffer Shuffle Algorithm}
So, ideally, we would like a different oblivious shuffle algorithm,
whose goal is still to obliviously permute the collection, $A$, of $N$
values, but with a simpler algorithm. 
The \emph{buffer-shuffle} algorithm is such an alternative:
\begin{enumerate}
\item
Perform a scan of $A$, $M$ numbers at a time.
With each step, we
read in $M$ wrapped items from $A$,
each of the form $(k',(k,v))$,
and randomly permute them in Alice's local memory.
\item
We then generate a new random key, $k''$, for each such wrapped
item, $(k',(k,v))$, in this group,
and we output all those new key-value pairs back to the server. 
\item
We then repeat this operation
with the next $M$ numbers, and so on, until Alice has made a pass
through all the numbers in $A$. 
\end{enumerate}
Call this a single \emph{pass}.
After such a pass, we can view the new keys as being sorted at the 
server (as observed above, by the properties of the OS model).
Thus, we can perform
another pass over these new key-value
pairs, generating an even newer set of wrapped key-value pairs. (This
functionality is supported by range queries, for example, so there is
little overhead for the client in implementing each such pass.) 
Finally, we
repeat this process for some constant, $b$, times, which is established
in our analysis below. 
This is the buffer-shuffle algorithm.

\subsection{Buffer-Shuffle Analysis}
To analyze the buffer-shuffle algorithm, we first focus on the
following goal: we show that with probability $1-o(1)$ that after four
passes, one cannot guess the location of an initial key-value pair
with probability greater than ${1/N} + o(1/N)$, assuming $M=N^{1/3}$,
where $N$ is the number of items being permuted.  After we prove this,
we discuss how the proof extends to obtain improved probabilities of
success and tighter bounds on the probability of tracking so that they
are closer to $1/N$, as well as how to extend to cases where
$M=N^{1/k}$ for integers $k \geq 3$.  

We think of the keys at the beginning of each pass as being in
key-sorted order, in $N^{2/3}$ groups of size $N^{1/3}$.  Let
$P_{i,j}$ be the perceived probability that after $i$ passes the key
we are tracking is in group $j$, according to the view of the tracker,
Bob.  Note that Bob can see, for each group on each pass, the set of
keys that correspond to that group at the beginning and end of the
pass, and use that to compute values $P_{i,j}$ corresponding to their
perceived probabilities.  Without loss of generality, we consider
tracking the first key, so $P_{0,1} = 1$.

Our goal will be to show that $P_{i,j} = N^{-2/3}+o(N^{-2/3})$,
for $i=3$ and for all $j$, conditioned on some events regarding
the random assignment of keys at each pass.  The events we condition
on will hold with probability $1-o(1)$.  This yields that the key
being tracked appears to a tracker to be (up to lower order terms) in
a group chosen uniformly at random. As the key values in each group
are randomized at the next pass, this will leave the tracker with
a probability only ${1/N} + o({1/N})$ of guessing the item,
again assuming the bad $o(1)$ events do not occur.

Let $X_{i,k,j}$ be the number of keys that go from group $k$ to group $j$ in pass $i$.
One can quickly check that $X_{i,k,j}$ is 0 with probability near 1.
Indeed, the probability that $X_{i,k,j} = c$ is bounded above by 
$${N^{1/3} \choose c} \left(N^{-2/3} \right )^c = O\left (N^{-c/3}\right ).$$

We have the recurrence
$$P_{i,j} = \sum_k P_{i-1,k} X_{i,k,j}/N^{1/3}.$$
The explanation for this recurrence is straightforward.
The probability the key being tracked is in the $j$th group in after pass $i$
is the sum over all groups $k$ of 
the probability the key was in group $k$, given by $P_{i-1,k}$, times
the probability the corresponding new key was mapped to group $j$, which 
$X_{i,k,j}/N^{1/3}$.

Our goal now is to show that over successive passes that as long as the values
$X_{i,k,j}$ behave nicely, the $P_{i,j}$ will quickly converge to roughly $N^{-2/3}$.  
We sketch an argument that with probability $1-o(1)$ and then comment on how the
$o(1)$ term can be reduced to any inverse polynomial probability in a constant number
of passes.  Our main approach is to note that bounding the $X_{i,k,j}$ corresponds
to a type of balls and bins problem, in which case negative dependence can be applied
to get a suitable concentration result via a basic Chernoff bound.  

\begin{theorem}  When $M = N^{1/3}$, after four passes,   
Bob cannot guess the location of an initial key-value pair
with probability greater than ${1/N} + o(1/N)$.
\end{theorem}

\begin{proof}
We consider passes in succession.

\begin{itemize}
\item Pass 1: It is easy to check that with probability $1-o(1)$
(using just union bounds and the binomial distribution to bound the
number of keys from group 1 that land in every other group) there are
at most $c \log N$  groups for which $X_{1,1,j} = 2$ and 0 groups for
which $X_{1,1,j} = 3$.  There are therefore $N^{1/3} - O(\log N)$
groups for which $P_{1,j} = N^{-1/3}$ and $O(\log N)$ groups for which
$P_{1,j} = 2N^{-1/3}$.
\item Pass 2: Our interpretation here (and going forward) is that each
key in group $j$ after pass $i-1$ has a "weight" $P_{i-1,j}/N^{1/3}$ that
it gives to the group it lands in in pass $i$;  the sum of weights in
a group then yields $P_{i,j}$.  

With this interpretation, with probability $1-o(1)$, there are
$N^{2/3} -o(N^{2/3})$ keys at the end of pass 1 with positive weight
(of either $N^{-2/3}$ or $2N^{-2/3}$).  These keys are rerandomized,
so at the end of pass 2, the number of keys with positive with in a
given bucket $j$ is expected to be constant, and again simple binomial
and union bounds imply it the maximum number of keys with positive
weight in any bucket is at most $c \log N$ with probability $1-o(1)$.
Indeed, one can further show at the end of pass 2 that the number of
groups $j$ with $P_{2,j} > 0$ must be at least $\Omega(N^{2/3})$ with
probability $1-o(1)$; this follows from the fact that, for example, if
$Y_j$ is the 0/1 random variable that represents whether a group $j$
received at least one weighted key, then $E[Y_j] > 1/2$, and the $Y_j$
are negatively associated, so Chernoff bounds apply.  (See, for
example, Chapter 3 of \cite{DubhashiPanconesi}.)  
\item Pass 3:  Conditioned on the $1-o(1)$ events from the first two passes,
at the end of the second pass there are $\Omega(N)$ keys with positive 
weight going into pass 3, and the possible weight values for each key are bounded by 
$(c' \log N)/N$ for some constant $c'$.  The expected weight for each group after pass
3 is obviously $N^{-2/3}$.  The weight of the keys within a group are negatively associated,
so we can apply a Chernoff bound to the weight associated with each group, noting that to
apply the Chernoff bound we should re-scale so the range of the weights is $[0,1]$.  
Consider the first group, and let $Z_i$ be the weight of the $i$th keys
in the first group (scaled by multiplying the weight by $N/(c' \log N)$).  Let $Z =\sum_{i=1}^{N^{2/3}} Z_i$.
Then
$$\Pr(|Z - N^{1/3}/(c' \log N)| \geq N^{1/5}| \leq e^{-2N^{1/15}}.$$
Or, rescaling back, the weight in the first group is within $N^{-4/5}/(c' \log N)$ of
$N^{-2/3}$ with high probability, and a union bound suffices to show that this the same
for all groups.  
\item Pass 4:  After pass 3, with probability $1-o(1)$ each key has weight $1/N + o(1/N)$,
and so after randomizing, assuming the events of probability $1-o(1)$ all hold, the probability
that any key is the original one being tracked from Bob's point of view is $1/N + o(1/N)$.
\end{itemize}
\end{proof}

\paragraph{Extending the argument}
We remark that the $o(1)$ failure probability can be reduced to any inverse polynomial 
by a combination of choosing constant $c$ and $c'$ to be sufficiently high, and/or repeating
passes a (constant) number of times to reduce the probability of the bad events.  (For example,
if pass 1, fails with probability $p$, repeating it $a$ times reduces the failure probability to $p^a$;
the $o(1)$ failure probabilities are all inverse polynomial in $N$ in the proof above.)

Similarly, one can ensure that the probability that any key is the
tracked key to $1/N + o(1/N^a)$ for any constant $a$ by increasing the
number of passes further, but still keeping the number of passes
constant.  Specifically, note that we have shown that after the first
four passes, with high probability the weight of each key bounded
between $N^{-1} - N^{-j}$ and $N^{-1} + N^{-j}$ for some $j>1$, and
the total key weight is 1.  We re-center the weights around $N^{-1} -
N^{-j}$ and multiply them by $N^{j-1}$; now the new reweighted weights
sum to 1.  We can now re-apply above the argument; after four passes
we know that reweighted weights for each key will again be $N^{-1} -
N^{-j}$ and $N^{-1} + N^{-j}$.  Undoing the rescaling, this means the
weights for the keys are now bounded $N^{-1} - N^{-j}$ and $N^{-1} +
N^{1-2j}$, and we can continue in this fashion to obtain the desired closeness
to $1/N$.  

Finally, we note that the assumption that we can read in $N^{1/3}$
key-value pairs is and assign them new random key values can be reduced
to $N^{1/j}$ pairs for any $j \geq 3$.  
We sketch the proof.  Each step, as in the original proof, holds
with probability $1-o(1)$.  

In this case we have $N^{(j-1)/j}$ groups.
In the first pass, the weight from the shuffling is
spread to $\Omega(N^{2/j})$ key-value pairs, following the same reasoning as for
Pass 1 above.  Indeed, we can continue this argument;  in the next pass,
there weight will spread to $\Omega(N^{3/j})$ key-value pairs, and so on,
until after $j-2$ passes there are
$\Omega(N^{(j-1)/j})$ keys with non-zero weight with high probability,
with one small modification in the analysis:  at each pass, we can ensure that 
each group has less than $j$ weighted keys with high probability.  

Then, following the same argument as in Pass 2 above, one can show
that after the following pass $\Omega(N)$ keys have non-zero weight,
and the maximum weight is bounded above by $(c' \log N)/N$ for some
constant $c'$.  Applying the Chernoff bound argument for Pass 3 above 
to the next pass we find that the weight within each of the $N^{(j-1)/j}$
groups is equal to $1/N + o(1/N)$ after this pass, and again this suffices
by the recurrence to show that at most one more pass is necessary for each 
key-value pair to have weight $1/N + o(1/N)$.

\section{A Square Root Solution}
\label{sec:square}
As is a common practice in ORAM simulation papers, starting with the work of 
Goldreich and Ostrovsky~\cite{go-spsor-96}, before we give our 
more sophisticated
solutions to the oblivious storage problem, we first give a simple
\emph{square-root} solution. 
Our general solution is an inductive extension of this solution,
so the square-root also serves to form a basis for this induction.

In this square-root solution, we assume $M\ge N^{1/2}$.
Thus, Alice has a local memory of size at 
least $N^{1/2}$, and she and Bob can
exchange a message of size up to at least $N^{1/2}$ in a single I/O.
In addition, we assume that this solution provides an API for performing
oblivious dictionary operations where every 
\textsf{get}$(k)$ or \textsf{put}$(k,v)$ operation is guaranteed to 
be for a key $k$ that is contained in the set, $S$, 
that Alice is outsourcing to Bob.
That is, we give a miss-intolerant solution to the oblivious storage problem.

Our solution is based on the observation that we can view 
Alice's internal memory as a miss-tolerant solution to the OS problem.
That is,
Alice can store $O(M)$ items in her private 
memory in some dictionary data structure,
and each time she queries her memory for a key $k$ 
she can determine if $k$ is present without
leaking any data-dependent information
to Bob.

\subsection{The Construction}
Let us assume we have a miss-tolerant dictionary, $D(N)$,
that provides a solution to the OS problem that works
for sets up to size $N$, with at most $O(1)$ amortized number of I/Os 
of size at most $N$ per access.
Certainly, a dictionary stored in Alice's internal memory suffices
for this purpose (and it, in fact, doesn't even need any I/Os per access),
for the case when $N$ is at most $M$, the size of Alice's internal memory.

The memory organization of our solution, $B(N)$,
we describe here,
consists of two caches:
\begin{itemize}
\item
A cache, $C_0$, which is of size $M$ and is implemented using 
an instance of a $D(M)$ solution.
\item
A cache, $C_1$, which is of size $N+M$, which is stored 
as a dictionary of key-value pairs using Bob's storage.
\end{itemize}

The extra $M$ space in $C_1$ is for storing $M$ 
``dummy'' items, which have keys indexed 
from a range that is outside of the universe used for $S$, 
which we denote as
$-1,-2,\ldots,-M$.
Let $S'$ denote the set of $N$ items from $S$, plus items 
with these dummy keys (along with null values), minus any items in $C_0$.
Initially, $C_0$ is empty and $C_1$ stores the entire set $S$ plus the $M$
items with dummy keys.
For the sake of obliviousness,
each item, $(k,v)$, in the set $S'$ 
is mapped to a substitute key by a nonce
pseudo-random hash function, $h_r$,
where $r$ is a random number chosen at the time Alice asks Bob to build (or
rebuild) his dictionary.
In addition, each value $v$ is encrypted as $E_K(v)$,
with a secret key, $K$, known only to Alice.
Thus, each item $(k,v)$ in $S'$ is stored by Bob as 
the key-value pair $(h_r(k),E_K(v))$.

To perform an access, either for a \textsf{get}$(k)$ or \textsf{put}$(k,v)$,
Alice first performs a lookup for $k$ in $C_0$, using its technology
for achieving obliviousness. If she 
does not find an item with key $k$ (as she won't initially), then she
requests the item from Bob by issuing a 
request, \textsf{get}$(h_r(k))$, to him.
Note that, since $k$ is a key in $S$, and it is not in Alice's cache,
$h_r(k)$ is a key in $S'$, by the fact that we are constructing 
a miss-intolerant OS solution.
Thus, there will be an item returned from this request.
From this returned item, $(h_r(k),E_K(v))$,
Alice decrypts the value, $v$, and stores the item $(k,v)$ in
$C_0$, possibly changing $v$ if she is performing a \textsf{put} operation.
Then she asks Bob to remove the item with key $h_r(k)$ from $S'$.

If, on the other hand,
in performing an access for a key $k$, Alice finds a matching item for $k$ 
in $C_0$, then she uses
that item and she issues a 
dummy request to Bob by asking him to perform a \textsf{get}$(h_r(-j))$
operation, where $j$ is a counter she keeps in her local memory for the next 
dummy key.
In this case, she inserts this dummy item into $C_0$ and 
she asks Bob to remove the item with key $h_r(-j)$ from $S'$.
Therefore, from Bob's perspective, Alice is always requesting a random key
for an item in $S'$ and then immediately removing that item.
Indeed, her behavior is always that of doing a \textsf{get} from $C_0$, 
a \textsf{get} from $C_1$, a \textsf{remove} from $C_1$, and then
a \textsf{put} in $C_0$.

After Alice has performed $M$ accesses, $C_0$ will be holding
$M$ items, which is its capacity.
So she pauses her performance of accesses at this time and enters a
\emph{rebuilding} phase.
In this phase, she rebuilds a new version of the dictionary
that is being maintained by Bob.

The new set to be maintained by Bob is the current $S'$ unioned with 
the items in $C_0$ (including the dummy items).
So Alice resets her counter, $j$, back to $1$.
She then performs an oblivious shuffle of the set $S''=C_0\cup S'$.
This oblivious shuffle is performed either with an
external-memory sorting algorithm~\cite{gm-paodor-11} or with the buffer-shuffle
method described above, depending, respectively, on whether Alice desires
perfect obscurity or if she can tolerate a small amount of information
leakage, as quantified above.
Finally, after this random shuffle completes,
Alice chooses a new random nonce, $r$,
for her pseudo-random function, $h_r$.
She then makes one more pass over the set of items (which are masked and
encrypted) that are now stored by Bob (using
\textsf{getRange} operations as in the buffer-shuffle method), and she maps
each item $(k,v)$ to the pair $(h_r(k),E_K(v))$ and asks Bob to store this
item in his memory.
This begins a new ``epoch'' for Alice to then use for the next $M$
accesses that she needs to make.

Let us consider an amortized analysis of this solution.
For the sake of amortization, we charge each of the previous $M$
accesses for the effort in performing a rebuild.
Since such a rebuild takes $O(N/M)$ I/Os, provided $M\ge N^{1/c}$,
for some constant $c\ge 2$, this means we will
charge $O(N/M^2)$ I/Os to each of these previous accesses.
Thus, we have the following.

\begin{lemma}
\label{lem:square}
Suppose we are
given a miss-tolerant OS solution, $D(N)$, which achieves $O(1)$ amortized
I/Os per access for messages of up to size $M$,
when applied to a set of size $N/M$.
Then we can use this as a component, $D(N/M)$, of a 
miss-intolerant OS solution, $B(N)$,
that achieves $O(1)$ amortized I/Os per access for messages of size up
to $M\ge N^{1/c}$, for some constant $c\ge2$.
The private memory required for this solution is $O(M)$.
\end{lemma}

\begin{proof}
The number of amortized I/Os will be $O(1)$ per access, from $D(N)$.
The total number of I/Os needed to do a rebuild is
$O(N/M)$, assuming $M$ is at least $N^{1/c}$, 
for some constant $c\ge 2$.
There will be $N/M$ 
items that are moved in this case, which is equal to the number 
of previous accesses; hence, the amortized number of I/Os 
will be $O(1)$ per access.
The performance bounds follow immediately from the above discussion and the
simple charging scheme we used for the sake of an amortized analysis.
For the proof of security, note that each access that Alice makes to the
dictionary $S'$ will either be for a real item or a dummy element.
Either way, Alice will make exactly $N/M$
requests before she rebuilds this dictionary stored with Bob.
Moreover, from the adversary's perspective, every request
is to an independent uniformly random key, which is then immediately removed
and never accessed again.
Therefore, the adversary cannot distinguish 
between actual requests and dummy requests. 
In addition, he cannot correlate 
any request from a previous epoch, since Alice randomly shuffles the set of
items and uses a new pseudo-random function with each epoch.
\end{proof}

By then choosing $M$ appropriately, we have the following.

\begin{theorem}
\label{thm:square}
The square-root solution achieves $O(1)$ amortized I/Os for each data access,
allowing a client, Alice, to obliviously store $N$ items 
in a miss-intolerant way
with an honest-but-curious server, Bob, using 
messages that are of size at most $M=N^{1/2}$
and local memory that is of size at least $2M$.
The probability that this simulation fails to be oblivious is exponentially
small for a polynomial-length access sequence, if oblivious sorting is used
for shuffling, and polynomially small if buffer-shuffling is used.
\end{theorem}
\begin{proof}
Plugging $M=N^{1/2}$ into Lemma~\ref{lem:square} gives us the complexity bound. 
The obliviousness follows from the fact that if she has an internal memory
of size $M=N^{1/2}$, then Alice can easily implement a
miss-tolerant OS solution in her internal memory, which achieves the
conditions of the $D(M)$ solution
needed for the cache $C_0$.
\end{proof}

Note that the constant factor in the amortized I/O overhead
in the square-root solution is quite
small. 


Note, in addition, that by the obliviousness definition in the OS model, it
does not matter how many accesses Alice makes to the solution, $B(N)$,
provided that her number of accesses are not self-revealing of 
her data items themselves.\footnote{An access sequence would be
	self-revealing, for example, if Alice reads a value and then performs a
	number of accesses equal to this value.}

\section{Miss-Tolerance} \label{sec:miss-tolerance}
An important functionality that is lacking from the square-root solution is
that it does not allow for accesses to items that are not in the set~$S$.
That is, it is a miss-intolerant OS solution.
Nevertheless, we can leverage the square-root solution to allow for such
accesses in an oblivious way, by using a hashing scheme.


\subsection{Review of Cuckoo Hashing}
The main idea behind this extension is to use a miss-intolerant
solution to obliviously implement a \emph{cuckoo hashing}
scheme~\cite{pr-ch-04}.  In cuckoo hashing, we have two hash tables $T_1$
and $T_2$ and two associated pseudo-random hash functions, $f_1$
and~$f_2$.  An item $(k,v)$ is stored at $T_1[f_1(k)]$ or
$T_2[f_2(k)]$.  When inserting item $(k,v)$, we add it
to~$T_1[f_1(k)]$.  If that cell is occupied by another item,
$(\hat{k},\hat{v})$, we evict that item and place it in~$T_2[f_2(\hat{k})]$.  Again,
we may need to evict an item.  This sequence of evictions continues
until we put an item into a previously-empty cell or we detect an
infinite loop (in which case we rehash all the items).  Cuckoo hashing
achieves $O(1)$ expected time for all operations with high
probability.  This probability can be boosted even higher to $1-1/n^s$
by using a small cache, known as a \emph{stash}~\cite{kmw-chs-09}, of size
$s$ to hold items that would have otherwise caused infinite insertion
loops.  With some additional effort (e.g., see~\cite{arbitman2009amortized}), 
cuckoo
hashing can be de-amortized to achieve $O(1)$ memory
accesses, with very high probability, for insert, remove, and lookup
operations.

In most real-world OS solutions, standard cuckoo hashing should
suffice for our purposes. But, to avoid inadvertent data leakage and
ensure high-probability performance bounds, let us assume we will be
using de-amortized cuckoo hashing.

\subsection{Implementing Cuckoo Hashing with a Miss-Intolerant OS Solution}
Let us assume we have a miss-intolerant solution, $B(N)$, to 
the OS problem, which achieves a constant I/O complexity for accesses, using
messages of size~$M$.

A standard or de-amortized cuckoo hashing scheme provides an interface for
performing \textsf{get}$(k)$ and \textsf{put}$(k,v)$ operations, so that
\textsf{get} operations are miss-tolerant.
These operations are implemented using pseudo-random hash functions in the
random access memory (RAM) model, i.e., using a collection of memory cells,
where each such cell is uniquely identified with an index~$i$.
To implement such a scheme using solution $B(N)$, 
we simulate a read of cell $i$ 
with \textsf{get}$(i)$ operation and we simulate a write of $x$ to cell $i$
with \textsf{put}$(i,x)$.
Thus, each access using $B(N)$ is guaranteed to return an item, namely a cell
$(i,x)$ in the memory (tables and variables) 
used to implement the cuckoo-hashing scheme. 
Thus, whenever we access a cell with index
$i$, we actually perform a request for (an encryption of) this cell's
contents using the obliviousness mechanism provided by~$B(N)$.

That is,
to implement a standard or de-amortized cuckoo hashing scheme using $B(N)$, 
we assume now that every (non-dummy) key in Alice's simulation 
is an index in the memory used to implement the hashing scheme.
Thus, each access is guaranteed to return an item. 
Moreover,
because inserts, removals, and lookups achieve a constant number of memory
accesses, with very high probability, in a de-amortized cuckoo-hashing 
scheme (or with constant expected-time performance in a standard
cuckoo hashing scheme),
then each operation in a simulation of de-amortized cuckoo hashing
in $B(N)$ involves
a constant number of accesses with very high probability.
Therefore, using a de-amortized cuckoo-hashing scheme, we have the following result.

\begin{theorem} \label{thm:cuckoo}%
  Given a miss-intolerant OS solution, $B(N)$, that achieves $O(1)$
  amortized I/O performance with messages of size $M$ and achieves
  confidentiality and hardness of correlation,
we can implement a miss-tolerant solution, $D(N)$, that
  achieves $O(1)$ amortized I/O performance and 
also achieves confidentiality and hardness of correlation.
\end{theorem}

A standard cuckoo-hashing scheme yields instead the following result.

\begin{theorem}
  Given a miss-intolerant OS solution, $B(N)$, that achieves expected
  $O(1)$ amortized I/O performance, with messages of size $M$, we can
  implement a miss-tolerant solution, $D(N)$, that achieves $O(1)$
  expected amortized I/O performance.
\end{theorem}

Our use of cuckoo hashing in the above construction
is quite different, by the way, than previous 
uses of cuckoo-hashing for oblivious 
RAM simulation~\cite{gm-paodor-11,shornbs-klo-11,futurepaper}.
In these other papers, the server, Bob, gets to see the
actual indexes and memory addresses used in the cuckoo hashing scheme. 
Thus, the adversary in these other schemes can see where items 
are placed in cuckoo tables (unless their construction is itself
oblivious) and when and where they are 
removed; hence, special care must be taken to construct 
and use the cuckoo tables in an oblivious way.
In our scheme, the locations in the cuckoo-hashing scheme
are instead obfuscated because they are themselves built on
top of an OS solution. 

Also, in previous schemes, cuckoo tables were chosen for the reason
that, once items are inserted, their locations are determined by
pseudo-random functions.  Here, cuckoo tables are used only for the
fact that they have constant-time insert, remove, and lookup
operations, which holds with very high probability for de-amortized
cuckoo tables and as an expected-time bound for standard cuckoo tables.


\section{An Inductive Solution}
\label{sec:induction}
The miss-tolerant square-root method given in Section~\ref{sec:miss-tolerance}
provides a solution of the oblivious
storage problem with amortized constant I/O performance
for each access, but 
requires  Alice to have a local memory of size $2N^{1/2}$
and messages to 
be of size $N^{1/2}$ during the rebuilding phase (although 
constant-size messages are exchanged during the access phase).
In this section, we show how to recursively apply this method to create a
more efficient solution.

For an integer $c \geq 2$, let $D_c(N)$ denote a miss-intolerant
oblivious storage solution that has the following properties:
\begin{enumerate}
\item
It supports a dictionary of $N$ items.
\item
It requires local memory of size $cN^{1/c}$ at the client.
\item
It uses messages of size $N^{1/c}$.
\item It executes $O(1)$ amortized I/Os per operation (each
  \textsf{get} or \textsf{put}), where the constant factor in this bound
  depends on the constant $c$.
\item It achieves confidentiality and hardness of correlation.
\end{enumerate}
Note that using this notation, the square-root method derived in
Section~\ref{sec:miss-tolerance} using cuckoo hashing is a $D_2(N)$
oblivious storage solution.

\subsection{The Inductive Construction}
For our inductive construction, for $c\ge 3$,
we assume the existence of a 
oblivious storage solution $D_{c-1}(N')$.
We can use this
to build a miss-tolerant oblivious storage solution, $D_c(N)$, using
message size $M=N^{1/c}$ as follows:
\begin{enumerate}
\item
Use the construction of Lemma~\ref{lem:square}
to build a miss-intolerant OS solution, $B_c(N)$, from $D_{c-1}(N/M)$.
This solution will have $O(1)$ amortized I/Os per access,
with very high probability,
using messages of size $M$ and private memory requirement
of size $O(M)$ since $D_{c-1}(N/M)$ uses memory of size
\[
(N/M)^{\frac{1}{c-1}} = N^{\frac{1-1/c}{c-1}} = N^{1/c} = M.
\]
\item
Use the construction of Theorem~\ref{thm:cuckoo} to take 
the miss-intolerant solution, $B_c(N)$, and convert
it to a miss-tolerant solution.
This solution uses an $O(1)$ amortized
number of I/Os, with high probability, using messages of size $M=N^{1/c}$,
and it has the performance bounds necessary to be denoted as $D_c(N)$.
\end{enumerate}

An intuition of our construction is as follows.
We number each level of the construction such that $c$ is the top most
and $2$ is the lowest level, hence there are $c-1$ levels.
The top level, $c$, consists
of  the main memory $A_c$ of size $O(N)$ and uses the rest of the construction
as a cache for $O(N/M)$ items which we
referred to as $B_c(N)$. This cache is the beginning of our
inductive construction and, hence, itself is an OS over $O(N/M)$ items.
The inductive construction continues such that level $i$ contains a
miss-tolerant data structure $A_i$ and levels $(i-1), \ldots, 2$
are used as a cache of level $i$.
The construction terminates when we reach level 2 since size of the
cache at level 2 is equal to the message size
$M$ which Alice can request using a single access or store in her own memory.
We give an illustration of our construction for $c=2$ and $c=3$ in
Figures~\ref{fig:ind_construct_k2} and~\ref{fig:ind_construct_k3}.

\ifFull
\begin{figure}[t]
\else
\begin{figure}[ht!]
\fi
\begin{center}
\includegraphics[scale=0.45]{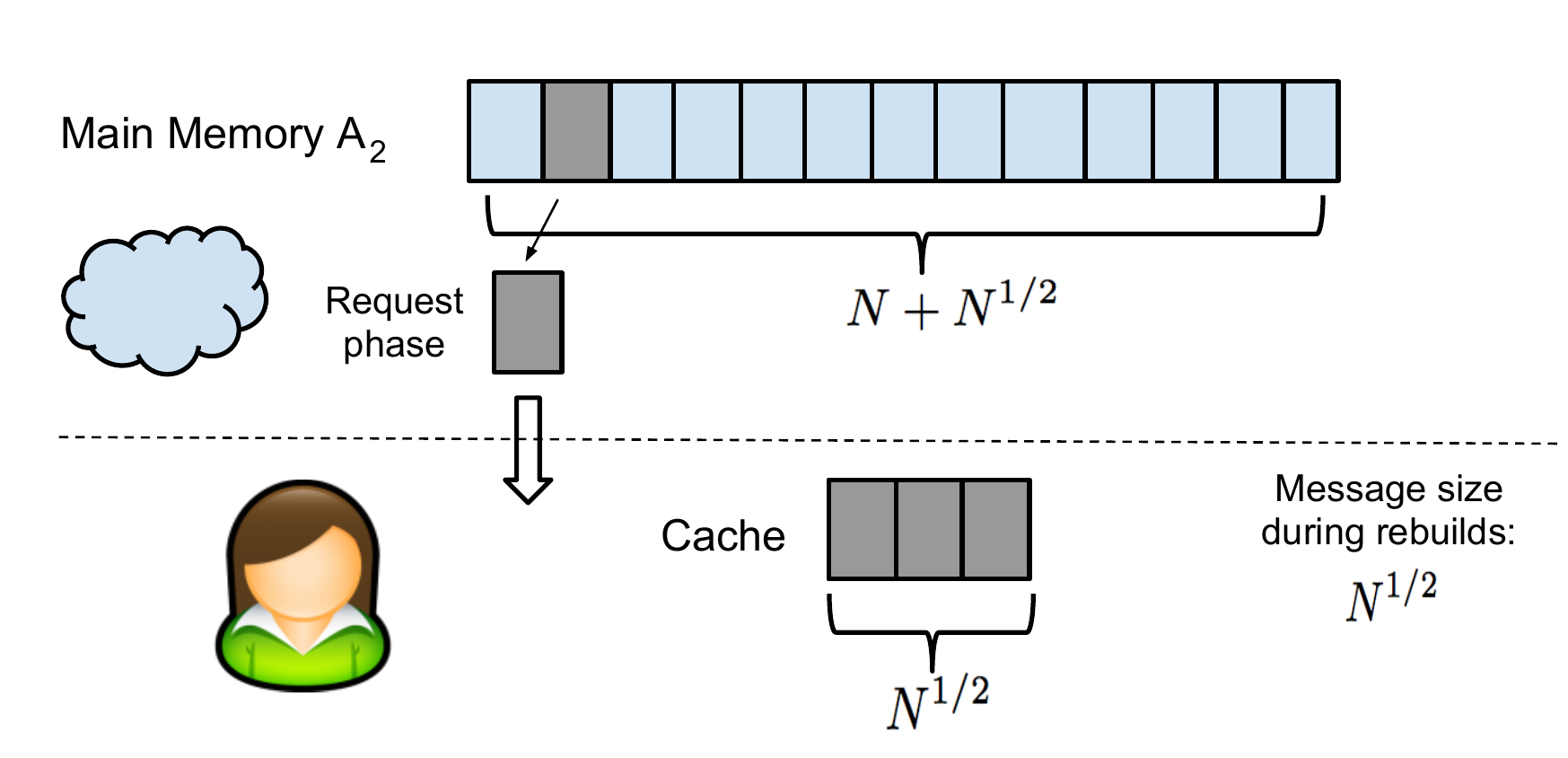}
\caption{Memory layout for $c=2$.
The locations accessed by the user are visualized as gray-filled rectangles.}
\label{fig:ind_construct_k2}
\end{center}
\end{figure}

\ifFull
\begin{figure}[t]
\else
\begin{figure}[ht!]
\fi
\begin{center}
\includegraphics[scale=0.45]{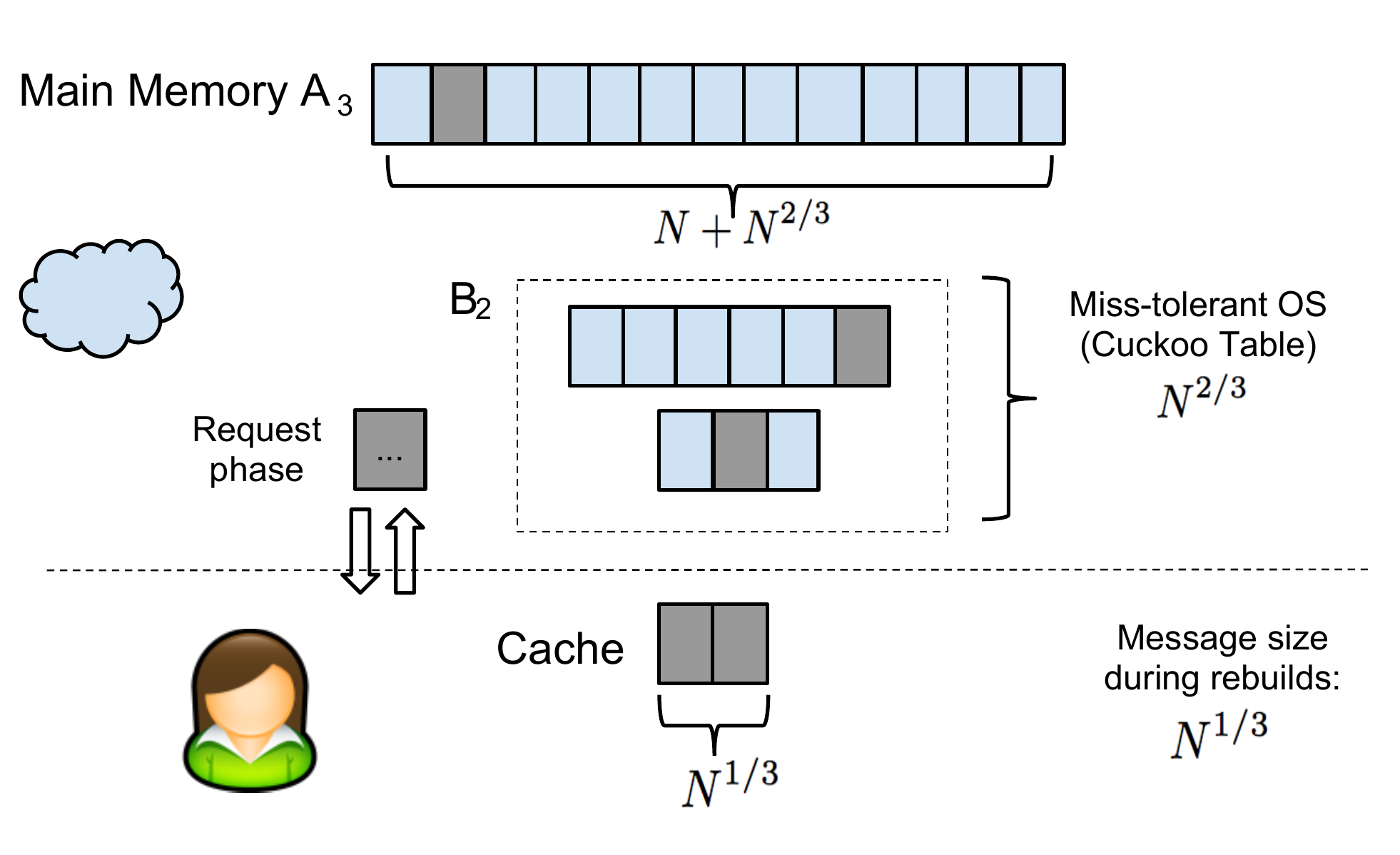}
\caption{Memory layout for $c=3$.}
\label{fig:ind_construct_k3}
\end{center}
\end{figure}

\begin{theorem}
The above construction results in an oblivious storage solution,
$D_c(N)$, that is miss-intolerant,
supports a dictionary of $N$ items,
requires client-side local memory of size at least $cN^{1/c}$,
uses messages of size $N^{1/c}$,
achieves an amortized $O(1)$ number of I/Os for each
\textsf{get} and \textsf{put} operation, 
where the constant factor in this bound depends on 
the constant $c\ge 2$. In addition, this method achieves confidentiality 
and hardness of correlation.
\end{theorem}


\section{Performance}
\label{sec:performance}

We have built a system prototype of our oblivious storage method to
estimate the practical performance of our solution and compare it with
that of other OS solutions.  In our simulation, we record the number
of access operations to the storage server for every original data
request by the client.
Our prototype specifically simulates the use of Amazon S3 as the
provider of remote storage, based on their current~API.  In
particular, we make use of operations \textsf{get}, \textsf{put},
\textsf{copy} and \textsf{delete} in the Amazon S3 API.  Since Amazon
S3 does not support range queries, we substitute operation
\textsf{getRange}$(i_1,i_2,m)$ of our OS model with $m$ concurrent
\textsf{get} requests, which could be issued by parallel threads
running at the client to minimize latency.  Operation
\textsf{removeRange} is handled similarly with concurrent
\textsf{delete} operations.
We have run the simulation for two configurations of our OS solution,
$c=2$ and~$c=3$. We consider two item sizes, 1KB and 64KB.  The size
(number of items) of the messages exchanged by the client and server
is $M=N^{1/c}$ where $N$, the number of items in the outsourced data
set, varies from $10^4$ to~$10^6$.

\textbf{Storage overhead.} The overall storage space (no.\ of items) used
by our solution on the server is $N+2\sum_{i=1}^{c-2}N^{(c-i)/c}$,
i.e.  $N+N^{1/2}$ for $c=2$ and $N+2N^{2/3}$ for $c=3$. For $c=2$, our
method has storage overhead comparable to that of Boneh {\it et
  al.}~\cite{bmp-rosmor-11} and much smaller than the space used by
other approaches.

\textbf{Access overhead.}  In Table~\ref{tbl:access_overhead}, we show
the number of I/Os to the remote data repository during the oblivious
simulation of $N$ requests.  Recall that the number of I/Os is the
number of roundtrips the simulation makes.  Thus, the
\textsf{getRange} operation is counted as one I/O.  In the table,
column Minimum gives the number of I/Os performed by Alice to receive the
requested item. The remaining I/Os are performed for reshuffling.  For
$c=2$, this number is 2 since the client sends a \textsf{get} request
to either get an actual item or a dummy item, followed by a
\textsf{delete} request. For $c=3$, this number is slightly higher
since Alice needs to simulate accesses to a cuckoo table through an OS
interface.  We compare our I/O overhead and the total amount of data
transferred with that of Boneh {\it et al.}~\cite{bmp-rosmor-11}.
They also achieve $O(1)$ request overhead and exchange messages of
size $M=N^{1/2}$ with the server.   Our new buffer shuffle
algorithm makes our approach more efficient in terms of data transfer
and the number of operations the user makes to the server.

\textbf{Time overhead.} Given the trace of user's operations during
the simulation and empirical measurements of round trip times of
operations on the Amazon S3 system (see Table~\ref{tbl:amazon_ops}),
we estimate the access latency of our OS solutions in
Tables~\ref{tbl:times_cost1kb} and \ref{tbl:times_cost64kb} for 1KB
items and 64KB items, respectively.

\textbf{Cost overhead.}  Finally, we provide estimates of the monetary
cost of OS our solution in Table~\ref{tbl:times_cost1kb} using the
pricing scheme of Amazon S3 (see Table~\ref{tbl:amazon_ops} and
\url{http://aws.amazon.com/s3/pricing/}~\footnote{Accessed on
  9/21/2011}).
Since our results outperform other approaches in terms of number of
accesses to the server we expect that our monetary cost will be also lower.

\begin{table*}
\ifFull
\small
\fi
\begin{center}
\begin{tabular}{r|c|c}
& \multicolumn{2}{c}{Minimum/Amortized} \\
\hline
$N$ & $M=N^{1/2}$ & $M=N^{1/3}$\\
\hline
10,000 & 2/13 & 7/173 \\
100,000 & 2/13 & 7/330 \\
1,000,000 & 2/13 & 7/416  \\
\end{tabular}
\end{center}
\caption{Minimum and amortized number of I/Os to access one item in our OS solution for $c=2$ and $c=3$.   We simulate a sequence of $N$ accesses on a system that uses four passes of the buffer shuffle algorithm.}
\label{tbl:access_overhead}
\begin{center}
\begin{tabular}{r||c|c||c|c}
\multirow{3}{*}{$N$} &\multicolumn{4}{c}{I/O Overhead and Data Transferred (\#items)} \\
\cline{2-5}
& \multicolumn{2}{c||}{Boneh {\it et al.}~\cite{bmp-rosmor-11}} & \multicolumn{2}{c}{Our Method} \\
\cline{2-5}
& Minimum & Amortized & Minimum & Amortized \\
\hline
10,000 & 3/9 & 13/1.3$\times10^3$ & 2/1  & 13/1.1$\times10^3$\\
100,000& 3/10 &17/5.2$\times10^3$ & 2/1  & 13/3.5$\times10^3$\\
1,000,000  & 3/12 & 20/2$\times10^4$ & 2/1  & 13/1.1$\times10^4$ \\
\end{tabular}
\end{center}
\caption{Minimum and amortized number of I/Os and number of items transferred to access one item. We compare our OS solution for $c=2$ with that of~\cite{bmp-rosmor-11} on a data set with 1KB items. In both solutions the message size is  $M=N^{1/2}$ items.}
\begin{center}
\begin{tabular}{r|c|c|c}
\multirow{2}{*}{Operation} & \multirow{2}{*}{Price} &  \multicolumn{2}{c}{RTT (ms)}  \\
\cline{3-4}
& & 1KB & 64KB \\
\hline
Get & \$0.01/10,000req & 36 & 56 \\
Put & \$0.01/1,000req & 65 & 86 \\
Copy & free &  70 & 88 \\
Delete & free & 31 & 35 \\
\end{tabular}
\end{center}
\caption{Amazon S3's pricing scheme and empirical measurement of round-trip time (RTT) for an operation issued by a client in Providence, Rhode Island to the Amazon S3 service (average of 300 runs).}
\label{tbl:amazon_ops}
\begin{center}
\begin{tabular}{r|c|c|r|c|c|r}
\multirow{3}{*}{$N$}  & \multicolumn{3}{c|}{$M=N^{1/2}$} & \multicolumn{3}{c}{$M=N^{1/3}$} \\
\cline{2-7}
& \multicolumn{2}{c|}{Access Time} & \multirow{2}{*}{Total Cost} & \multicolumn{2}{c|}{Access Time} & \multirow{2}{*}{Total Cost}   \\
\cline{2-3} \cline{5-6}
& Minimum & Amortized & & Minimum & Amortized & \\
\hline
10,000 & 67ms & 500ms & \$55 & 400ms & 8s & \$177 \\
100,000 & 67ms & 500ms & \$1,744 & 400ms & 12s & \$7,262 \\
1,000,000 & 67ms & 500ms & \$55,066 & 400ms & 18s & \$170,646
\end{tabular}
\end{center}
\caption{Estimate of the access time per item and total monetary cost
for accessing $N$ items, each of size 1KB, stored on the Amazon S3 system using our OS method for $c=2$ and $c=3$.}
\label{tbl:times_cost1kb}
\begin{center}
\begin{tabular}{r|c|c|r|c|c|r}
\multirow{3}{*}{$N$}  & \multicolumn{3}{c|}{$M=N^{1/2}$} & \multicolumn{3}{c}{$M=N^{1/3}$} \\
\cline{2-7}
& \multicolumn{2}{c|}{Access Time} & \multirow{2}{*}{Total Cost} & \multicolumn{2}{c|}{Access Time} & \multirow{2}{*}{Total Cost}   \\
\cline{2-3} \cline{5-6}
& Minimum & Amortized & & Minimum & Amortized & \\
\hline
10,000 & 91ms & 800ms & \$55 & 500ms & 12s & \$177 \\
100,000 & 91ms & 800ms & \$1,744 & 500ms & 18s & \$7,262 \\
1,000,000 & 91ms & 800ms & \$55,066 & 500ms & 24s & \$170,646
\end{tabular}
\end{center}
\caption{Estimate of the access time per item and total monetary cost
for accessing $N$ items, each of size 64KB, stored on the Amazon S3 system using our OS method for $c=2$ and $c=3$.}
\label{tbl:times_cost64kb}
\end{table*}


\ifFull
\subsection*{Acknowledgments}
This research was supported in part by the National
  Science Foundation under grants
  0721491,
  0915922, 0953071, 0964473, 1011840, and 1012060,
  and by the Kanellakis Fellowship at
  Brown University.
\fi

\clearpage

\bibliographystyle{abbrv}
\bibliography{paper}

\end{document}